\documentclass[12pt]{article}

\usepackage{amsmath}
\usepackage{amssymb}
\usepackage{amsfonts}
\usepackage{latexsym}
\usepackage{color}
\usepackage{soul}
\usepackage{tikz}
\usetikzlibrary{arrows,decorations.markings,calc}
\usetikzlibrary{decorations.pathmorphing,patterns}
\usepackage{graphicx}
\usepackage{subfigure}
\usepackage{color,epstopdf}
\catcode `\@=11 \@addtoreset{equation}{section}

\catcode `\@=12

\DeclareFontFamily{U}{skulls}{}
\DeclareFontShape{U}{skulls}{m}{n}{ <-> skull }{}

\newcommand{\be}{\begin{equation}}
\newcommand{\en}{\end{equation}}
\newcommand{\bea}{\begin{eqnarray}}
\newcommand{\ena}{\end{eqnarray}}
\newcommand{\beano}{\begin{eqnarray*}}
\newcommand{\enano}{\end{eqnarray*}}
\newcommand{\bee}{\begin{enumerate}}
\newcommand{\ene}{\end{enumerate}}

\newcommand{\mc}{\mathcal}

\newcommand{\F}{{\cal F}}
\newcommand{\G}{{\cal G}}

\newcommand{\Lc}{{\cal L}}

\newcommand{\1}{1 \!\! 1}

\newcommand{\Hil}{\mc H}

\newtheorem{thm}{Theorem}

\newtheorem{prop}[thm]{Proposition}
\newtheorem{defn}[thm]{Definition}

\newenvironment{proof}{\noindent {\bf Proof --}}{\hfill$\square$ \vspace{3mm}\endtrivlist}

\begin{document}

\thispagestyle{empty}

\begin{center} {\Large \bf Density matrices and entropy operator for non-Hermitian quantum mechanics}   \vspace{1cm}\\

{\large F. Bagarello}\\
  Dipartimento di Ingegneria,
Universit\`a di Palermo, 90128  Palermo, Italy\\
and I.N.F.N., Sezione di Catania, 95123 Catania, Italy\\
e-mail: fabio.bagarello@unipa.it\\

\vspace{2mm}

{\large F. Gargano}\\
Dipartimento di Ingegneria,
Universit\`a di Palermo, 90128  Palermo, Italy\\
e-mail: francesco.gargano@unipa.it\\

\vspace{2mm}

{\large L. Saluto}\\
Dipartimento di Ingegneria,
Universit\`a di Palermo, 90128  Palermo, Italy\\
e-mail: lidia.saluto@unipa.it\\

\end{center}

\vspace{0.5cm}

\begin{abstract}
\noindent

In this paper we consider density matrices operator related to non-Hermitian Hamiltonians. In particular, we analyse two natural extensions of what is usually called a density matrix operator (DM), of pure states and of the entropy operator: we first consider those {\em operators} which are simply similar to a standard DM, and then we discuss those which are intertwined with a DM by a third, non invertible, operator, giving rise to waht we call Riesz Density Matrix operator (RDM). After introducing the mathematical framework, we apply the framework to a couple of applications. The first application is related to a non-Hermitian Hamiltonian describing gain and loss phenomena, widely considered in the context of $PT$-quantum mechanics. The second application is related to a finite-dimensional version of the Swanson Hamiltonian, never considered before, and addresses the problem of deriving a milder version of the RDM when exceptional points form in the system.

\end{abstract}

%{\bf PACS Numbers}:  .......

%\vfill

%\pagenumbering{roman}

\newpage

\section{Introduction}
In functional analysis, the analysis of Hilbert spaces and of the operators acting on them is quite important. This is true for mathematical reasons, of course, but also in view of their applications in quantum mechanics. Position, momentum, energy operators, as well as projection, translation, dilation operators, often play some role in the analysis of specific systems and, for this reason, they are very much studied in the literature. This is true also in the context of the recent version of quantum mechanics where self-adjointness of the observables is not necessarily required, \cite{specissue2012}-\cite{bagprocpa}. But many other operators may be relevant. This is the case of the so-called density matrices, whose role turns out to be particularly useful for open quantum systems, \cite{petru}. In the {\em standard} literature on quantum mechanics, a DM $\rho_0$ is, first of all, a self-adjoint bounded operator: $\rho_0=\rho_0^\dagger$. As we will discuss later, this implies that $\rho_0$ admits a set of eigenvectors which, under suitable assumptions, form an orthonormal basis (ONB) of the Hilbert space $\Hil$ where $\rho_0$ acts. But in the past few decades it becames clearer and clearer that ONB are not always the most natural set of vectors appearing when loosing self-adjointness. In many cases, one has to consider bi-orthogonal sets of vectors, which could be Riesz bases or not \cite{bagspringer,heil,chri}. In the literature, this passage from ONB to bi-orthogonal sets has been discussed by various authors, and under many different aspects. We only cite here \cite{bagspringer,brody,mosta}. What is not so considered, to our knowledge, is what are the changes for DMs. In other words: how should we define a DM in presence of a non self-adjoint Hamiltonian? These are indeed only few papers on this topic, as for instance \cite{brody,bagchi,harviou,sergi}, with only few information. For sure, what is still missing, is a general (abstract) treatment of this aspect of DMs. This is exactly what we are beginning here: a detailed analysis of what a DM can be thought to be for a quantum mechanical system driven by a non self-adjoint Hamiltonian. In particular, we will consider two different situations: in the first (and easiest) one, the {\em new} DM is simply similar to $\rho_0$, and the similarity is implemented by a bounded non unitary operator with bounded inverse. In this case, as one can easily imagine, bi-orthogonal Riesz bases will be relevant. This is the case mostly considered in the existing literature, see \cite{bagchi,harviou} in particular. However, we will also consider here the case in which the {\em new} DM $\rho$ is {\bf not} similar to $\rho_0$, but still $\rho$ and $\rho_0$ are linked by a certain intertwining operator, \cite{intop}. We will see that, in this case, the situation is much more delicate, but still interesting. We should also stress that the role of DMs is relevant also in connection with quantum mechanical states, pure or not, and with the definition of an entropy operator. These aspects will also be considered in our analysis, for our {\em extended} DMs.

The paper is organized as follows: in the next section, after a short review on DM, pure states and entropy operator for {\em ordinary} quantum mechanics, we extend these results to the cases where a similarity map exists, Section \ref{sect-sim}, and when it does not, Section \ref{sect-IO}. 
In Section \ref{sect-ex} we propose some examples
to validate our mathematical
framework. In particular, in Section \ref{sect-2D} we present a first application to a two-state system living in $\mathbb{C}^2$, while in Section \ref{sect-Swa} we introduce a sort of Swanson-like Hamiltonian where the usual bosonic ladder operators are replaced by a truncated version of the same operators, living in $\mathbb{C}^3$, or, with a  different view, as an extended version of ladder fermionic operators. For both models, some explicit examples of the DMs considered in  Section \ref{sect2} will be considered. In particular, for the extended Swanson model we will see that there is a difference between the unbroken and the broken phases when considering the entropy operator and its asymptotic behaviour.
Our conclusions are given in Section \ref{sect-concl}.

\section{Density matrices and pure states}\label{sect2}

The first part of this section is devoted to list some (well-known) facts on DMs and pure states in a {\em standard} settings, i.e. for self-adjoint DMs. Then we will extend these considerations to operators which are similar to  self-adjoint DMs. In the third part, we will consider the case in which a self-adjoint DM $\rho_0$ is related to a second operator $\rho$ via some intertwining operator, $V$: $\rho V=V\rho_0$. Of course, if $V^{-1}$ exists, we would go back to the previous situation, where a similarity relation exists between $\rho$ and $\rho_0$. Hence the interesting case will be that in which $V$ has no inverse.

Before we start, let us introduce some useful notation we will use all along the paper: we call $\Hil$ our Hilbert space, endowed with scalar product $\langle.,.\rangle$, and with related norm $\|.\|=\sqrt{\langle.,.\rangle}$. $\Hil$ could be finite or infinite-dimensional.  $B(\Hil)$ is the $C^*$-algebra of the bounded operator acting on $\Hil$. The adjoint $\dagger$ is the one fixed by the scalar product on $\Hil$:  $\langle A^\dagger f,g\rangle=\langle f,Ag\rangle$, for all $f,g\in\Hil$. Here $A\in B(\Hil)$.

\subsection{A short review for $\rho_0=\rho_0^\dagger$}

We begin with the following definition:

\begin{defn}\label{def1}
	An operator $\rho_0\in B(\Hil)$ is called a {\em density matrix}, DM, if $\rho_0>0$ and if $tr(\rho_0)=1$.
\end{defn}

It is worth remarking that our definition here differs from that one usually find in books on quantum mechanics, see \cite{merz,roman} or, more recently, \cite{brody} for instance, since in these latter the authors explicitly require also $\rho_0$ to be Hermitian\footnote{Here Hermitian and self-adjoint will often be used as synonymous.}. This is indeed redundant since the request that $\rho_0$ is positive automatically implies its Hermiticity, \cite{reed}. If $\rho_0$ is a DM, then $|\rho_0|=\sqrt{\rho^\dagger \rho}=\rho_0$, so that $tr|\rho_0|=tr(\rho_0)=1$. Hence $\rho_0\in{\cal T}_1$, the set of all the trace-class elements in $B(\Hil)$, \cite{reed}. It is known that ${\cal T}_1\subset Com(\Hil)$, the set of the compact operator on $\Hil$. Hence we can use the Hilbert-Schmidt theorem, \cite{reed}, which states that $\rho_0$ admits a set of eigenvalues $\{\lambda_j\}$ and an orthonormal basis (ONB) $\F_e=\{e_j\}$, such that
\be
\rho_0\,e_j=\lambda_j\,e_j,
\label{21}\en
where $\lambda_j\rightarrow0$, when $j\rightarrow\infty$. Then we can rewrite $\rho_0$ as follows
\be
\rho_0=\sum_j\lambda_j\,P_j^o,
\label{22}\en
where $P_j^o$ are orthogonal projectors acting as follows: $P_j^of=\langle e_j,f\rangle e_j$, $\forall f\in\Hil$, or, using a bra-ket language, $P_j^o=|e_j\rangle\langle e_j|$. To fix the ideas, we will assume here often that $j\in\mathbb{N}$. In order for $\rho_0$ to be a DM the sequence $\{\lambda_j\}$ must be such that
\be
\sum_j\lambda_j=1, \qquad\mbox{ and }\qquad \lambda_j\in[0,1], \quad \forall j.
\label{22b}\en

Since $\rho_0$ is positive, it admits an unique positive square root $\rho_0^{1/2}$, which belongs to ${\cal T}_2$, \cite{reed}. Moreover, since ${\cal T}_1$ is a two-sided ideal for $B(\Hil)$, it follows that $\rho_0 X, X\rho_0\in{\cal T}_1$, for all $X\in B(\Hil)$. For this reason we can introduce a well defined linear functional $\Phi_{\rho_0}$ on $B(\Hil)$ as follows:
\be
\Phi_{\rho_0}(X)=tr(\rho_0X).
\label{23}\en
$\Phi_{\rho_0}$ is linear, normalized, positive and continuous. More explicitly:
\be
\Phi_{\rho_0}(\alpha X+\beta Y)=\alpha\Phi_{\rho_0}(X)+\beta\Phi_{\rho_0}(Y), \qquad \Phi_{\rho_0}(\1)=1,
\label{24}\en
$\forall X,Y\in B(\Hil)$ and $\alpha,\beta\in\mathbb{C}$. Moreover, if $X\in B(\Hil)$ is positive, $X>0$, then $\Phi_{\rho_0}(X)>0$. Also, if $\|X_n-X\|\rightarrow0$, then $\Phi_{\rho_0}(X_n)\rightarrow\Phi_{\rho_0}(X)$.

The set of DMs, $\G$, is convex: if $\rho_1$ and $\rho_2$ are DMs, then $\rho=\lambda\rho_1+(1-\lambda)\rho_2$ is a DM for all $\lambda\in[0,1]$. Moreover, any DM $\rho$ defines an operator called {\em its entropy}. In particular, $\rho_0$ in (\ref{22}) produces
\be
S(\rho_0)=-\rho_0\,\log(\rho_0)=-\sum_j\lambda_j\,\log(\lambda_j)P_j^o.
\label{25}\en

\vspace{2mm}

{\bf Remarks:--} (1) $S(\rho_0)$ is clearly well defined if the sum is finite, and in this case it is also trivially a bounded operator. Due to the fact that each $\lambda_j$ belongs to the interval $[0,1]$, $\log(\lambda_j)\leq0$ for all $j$. Hence $-\sum_j\lambda_j\,\log(\lambda_j)\geq0$. Notice that, when $\lambda_j=0$, using a well known result, we define $\lambda_j\,\log(\lambda_j)=0 $, by continuity. If the set of $j$'s is infinite, the convergence of (\ref{25}) is more delicate. { For instance, it can be explicitly checked if $\lambda_j=(1-q)q^j$, for all possible $q\in]0,1[$, or if $\lambda_j=\frac{6}{(\pi(j+1))^2}$, $j=0,1,2,\ldots$.  More examples can also be easily constructed. What is seems not so easy is to set up a general proof of this convergence which, however, is not really essential for our purposes here.}

(2) In the literature the trace of $S(\rho_0)$ is usually called the {\em von Neumann entropy}.

\vspace{2mm}

All the normalized vectors $\Psi\in\Hil$ define a DM: calling $\rho_\Psi=P_\Psi=|\Psi\rangle\langle\Psi|$, the orthogonal projection operator associated to $\Psi$, then $\rho_\Psi$ is indeed a DM, as it is easily checked. Moreover, $\forall X\in B(\Hil)$,
$tr(\rho_\Psi X)=\langle\Psi,X\Psi\rangle$.

Among the DMs, a special class is that of so-called {\em pure states}: a DM $\rho_0$ is a pure state (or, maybe more properly, defines a pure state) if there is a normalized vector $\Phi_0\in\Hil$, $\|\Phi_0\|=1$, such that $\rho_0=|\Phi_0\rangle\langle\Phi_0|$.

The following theorem, which can be found in many references on DMs, provides a nice characterization of pure states:

\begin{thm}\label{thm1}
	A DM $\rho$ is a pure state if and only if one of the following properties, all equivalent, is satisfied:
	
	{\bf p1.} $tr(\rho^2)=1$.
	
	{\bf p2.} $S(\rho)=0$.
	
	{\bf p3.} $\rho$ is an extremal point of $\G$.

\end{thm}

\subsection{Similarity operators and DMs}\label{sect-sim}

The first natural extension of a DM is the one which is generated by a DM with the action of a bounded operator with bounded inverse. This is exactly what happens when going from orthonormal to Riesz bases, and in this sense it is a relevant situation both for mathematics, see \cite{heil,chri}, and for more physical situations, \cite{bagspringer}. What we will see here is that this extension is not entirely trivial, and produces several interesting results.

\begin{defn}\label{defRDM}
	Let $\rho_0$ be a DM and $R\in B(\Hil)$ invertible, with inverse in $B(\Hil)$. The operator\be
	\rho=R\rho_0 R^{-1}
	\label{26}\en
	is called an $(R,\rho_0)$-Riesz density matrix.
\end{defn}
Quite often, in the following, we will simply call $\rho$ a Riesz density matrix (RDM). In particular, this will be done whenever the role of $R$ and $\rho_0$ is clear. It is clear that the interesting situation is when $R$ is not unitary. In fact, if $R^\dagger=R^{-1}$, $\rho$ in (\ref{26}) shares with $\rho_0$ the same properties. Therefore, from now on, except if explicitly stated, we will work under the assumption that  $R^\dagger\neq R^{-1}$.

The first remark is that $\rho\in{\cal T}_1$. This is because $\rho_0\in{\cal T}_1$, which is an ideal for $B(\Hil)$. Now, if we define  two bi-orthonormal Riesz bases $\F_\varphi=\{\varphi_j=R e_j\}$ and  $\F_\psi=\{\psi_j=(R^{-1})^\dagger e_j\}$, in analogy with (\ref{22}) we can rewrite $\rho$ as follows:
\be
\rho=\sum_j\lambda_j\,P_j,
\label{27}\en
where $P_j$ are (non-orthogonal) projectors acting as follows: $P_jf=\langle \psi_j,f\rangle \varphi_j$, or, using a bra-ket language, $P_j=|\varphi_j\rangle\langle \psi_j|$. Using (\ref{26}) it is easy to check that $tr(\rho)=1$. However, it is quite easy to see that, in general, $\rho$ needs not being positive, or even self-adjoint. Indeed, it is sufficient to consider the following simple example:
$$
\rho_0=\frac{1}{5}\left(
\begin{array}{cc}
	2 & 1  \\
	1 & 3  \\
\end{array}
\right), \qquad R=\left(
\begin{array}{cc}
	1 & 2  \\
	1 & 3  \\
\end{array}
\right) \qquad \mbox{ with }\qquad  R^{-1}=\left(
\begin{array}{cc}
	3 & -2  \\
	-1 & 1  \\
\end{array}
\right).
$$
With these choices, we find
$
\rho=\left(
\begin{array}{cc}
	1 & -\frac{1}{5}  \\
	1 & 0  \\
\end{array}
\right).
$
It is clear that $\rho\neq\rho^\dagger$. If we further consider the vector $f=\left(
\begin{array}{c}
	\frac{2}{5}   \\
	-1  \\
\end{array}
\right)$ then $\langle f,\rho f\rangle=-\frac{4}{25}$. Hence $\rho$ is not positive. Of course this simple situation shows that $\rho$ in (\ref{27}) has not the same properties of $\rho_0$, expect the fact that they both have unit trace.

Going back to $P_j$, while it is clear that $P_jP_k=\delta_{j,k}P_j$, it is also clear that $P_j^\dagger=|\psi_j\rangle\langle \varphi_j|\neq P_j$, in general.

It is easy to check that the set of all the  $(R,\rho_0)$-RDMs, $\G(R,\rho_0)$, for $R$ fixed, is closed under convex combinations: if $\rho_1,\rho_2\in\G(R,\rho_0)$, then $\lambda\rho_1+(1-\lambda)\rho_2\in\G(R,\rho_0)$ as well, for all $\lambda\in[0,1]$. 

\vspace{2mm}

Given a RDM we can introduce a related linear functional as we did in (\ref{23}):
\be
\Phi_{\rho}(X)=tr(\rho X)=tr(R\rho_0R^{-1}X)=\Phi_{\rho_0}(X_R),
\label{28}\en
where we have introduced the short-hand notation $X_R=R^{-1}XR$, and where $\Phi_{\rho_0}$ is the state in (\ref{23}). $\Phi_{\rho}$ is not a state in the usual sense, \cite{brat}. In particular, while it is easy to check that $\Phi_\rho$ is linear, normalized (i.e. $\Phi_\rho(\1)=1$) and continuous, since $|\Phi_\rho(X)|\leq\|R\|\|R^{-1}\|\|X\|$, for all $X\in B(\Hil)$, it is also clear that if $X=X^\dagger$ then $\Phi_\rho(X)$ needs not to be real, and that if $X>0$ then $\Phi_\rho(X)$ needs not to be positive. This is because $X_R$ is not Hermitian (even if $X=X^\dagger$) and $X_R$ is not positive, even if $X>0$.

\vspace{2mm}

In Definition \ref{defRDM} our starting point is a DM, and a bounded operator $R$ with bounded inverse. With these ingredients we can define a RDM. In fact, this construction can be reversed: suppose we have a $\rho\in B(\Hil)$ such that its eigenvalues and eigenvectors  satisfy the following properties:
\begin{equation}\label{29}
	\rho\,\varphi_j=\lambda_j\varphi_j, \qquad \lambda_j\in[0,1], \quad \mbox{and}\quad \sum_j\lambda_j=1,
\end{equation}
and $\F_\varphi=\{\varphi_j\}$ is a Riesz basis. Then  we have the following result:

\begin{thm}
	Under the above assumption $\rho$ is a RDM. 
\end{thm}
\begin{proof}
	Since $\F_\varphi$ is a Riesz basis we know that an $R\in B(\Hil)$ exists, with $R^{-1}\in B(\Hil)$, and an ONB $\F_e=\{e_j\}$, such that $\varphi_j=Re_j$. We also know that $\F_\psi=\{\psi_j=(R^{-1})^\dagger e_j\}$ is another Riesz basis, bi-orthonormal to $\F_\varphi$. Formula (\ref{29}) produces now $\tilde\rho\,e_j=\lambda\,e_j$, where $\tilde\rho=R^{-1}\rho R$, which is obviously bounded. Now our claim follows from the fact that $\tilde\rho$ can be written as in (\ref{22}), $\tilde\rho=\rho_0=\sum_j\lambda_j\,P_j^o$, and from the relation between $e_j$, $\varphi_j$ and $\psi_j$.

\end{proof}

{\bf Remarks:--} (1) Using the notation of Definition \ref{defRDM}, we can say that $\rho$ is a $(R,\tilde\rho)$-RDM.

\vspace{1mm}

(2) It is interesting to observe that the assumption of having a bounded $\rho$ is not really needed here, since it follows from the eigenvalue equation $\tilde\rho\,e_j=\lambda\,e_j$. Indeed, let $A$ be a generic operator satisfying $A\,e_j=\lambda\,e_j$, for some sequence $\lambda_j\in[0,1]$. Hence, taken $f\in D(A^\dagger)$ (to be identified) we can write, using the Parseval identity for $\F_e$
$$
\|A^\dagger f\|^2=\sum_j|\langle A^\dagger f,e_j\rangle|^2=\sum_j|\langle f,A e_j\rangle|^2=\sum_j|\lambda_j|^2|\langle f, e_j\rangle|^2\leq
$$
$$
\leq \sum_j|\langle f, e_j\rangle|^2=\|f\|^2.
$$
Hence $A^\dagger$ is bounded on $D(A^\dagger)$, so that it can be extended to all $\Hil$, and it is still bounded, with $\|A^\dagger\|\leq1$. This implies that $\|A\|=\|A^\dagger\|\leq1$. Hence $A$ is also bounded. Going back to our original problem, we find that $\tilde\rho\in B(\Hil)$. But $\rho=R\tilde\rho R^{-1}$. Hence also $\rho$ is bounded.

\vspace{2mm}

It may be useful to observe that $\rho$ in (\ref{29}) is also associated to a second bounded operator, $\rho^\dagger$, with the same (real) eigenvalues and with eigenvectors which are exactly the vectors in $\F_\psi$:
\begin{equation}\label{29-bis}
	\rho^\dagger \psi_j=\lambda_j\psi_j, \qquad \rho^\dagger=\sum_j\lambda_j\,P_j^\dagger,
\end{equation} 
where $P_j^\dagger=|\psi_j\rangle\langle \varphi_j|$. All we have deduced for $\rho$, of course, can be simply restated for $\rho^\dagger$.

The easiest, and possibly more natural, way to introduce a {\em pure state} in our case is just to require that $\rho$ in Definition \ref{defRDM} is the image of a pure state, i.e. that $\rho_0$ in (\ref{26}) is a pure state. Stated differently, we have the following:
\begin{defn}\label{def5}
 An $(R,\rho_0)$-RDM is a {\em Riesz pure state} (RPS) if $\rho_0$ is a pure state, i.e. if it exists a normalized vector $\Phi_0$ such that $\rho_0=|\Phi_0\rangle\langle\Phi_0|$. In this case, calling $\varphi_0=R\Phi_0$ and $\psi_0=(R^{-1})^\dagger\Phi_0$ we can write
 \begin{equation}\label{212}
 	\rho=R\rho_0R^{-1}=|\varphi_0\rangle\langle\psi_0|.
 \end{equation}
\end{defn}
It is clear that $\langle\varphi_0,\psi_0\rangle=1$. In this case we have
\begin{equation}
	\Phi_{\rho}(X)=\Phi_{\rho_0}(X_R)=\langle\Phi_0,X_R\Phi_0\rangle=\langle\psi_0,X\varphi_0\rangle,
\end{equation}
for all $X\in B(\Hil)$. This formula shows that a RPS does not correspond to a mean value. Which is, of course, in agreement with the fact that $\Phi_\rho$ is not positive defined. This should be kept in mind since it implies that a pure RDM does not necessarily is of the form $|\eta\rangle\langle\eta|$, for $\eta\in\Hil$. 

Following the standard case, we further introduce the entropy operator for $\rho$ as follows:
\be
S(\rho)=R\,S(\rho_0)R^{-1},
\label{211}\en
which is bounded since it is the product of three bounded operators, at least if $S(\rho_0)$ is bounded, as e.g. in our examples.

Theorem \ref{thm1} can be restated here, slightly changed, and we have the following:

\begin{thm}\label{thm2}
	A RDM $\rho$ is a RPS if and only if one of the following equivalent properties is satisfied:
	
	{\bf p1$'$.} $tr(\rho^2)=1$.
	
	{\bf p2$'$.} $S(\rho)=0$.

\end{thm}

\begin{proof}
	First we observe that if $\rho$ is a RPS then (\ref{212}) implies that $\rho^2=\rho$. Hence we have $1=tr(\rho^2)=tr(\rho)$. Viceversa, if $\rho$ is a RDM such that $tr(\rho^2)=1$, then, since $tr(\rho^2)=tr(\rho_0^2)=1$. Hence $\rho_0$ is a PS, and $\rho$ is a RPS.
	
	As for {\bf p2$'$}, suppose $\rho$ is a RPS. Then (\ref{211}) implies that $S(\rho)=0$, since $\rho_0$ is a PS. Vice-versa, if $\rho$ is not pure, then $\rho_0$ is not pure, too. Then $S(\rho_0)\neq0$ and, see again (\ref{211}), $S(\rho)\neq0$ as well.

\end{proof}

\vspace{2mm}

{\bf Remark:--} We are not considering here the extremality of $\rho$, point {\bf p3} of Theorem \ref{thm1}, since it is not particularly useful for us, here.

\subsection{Intertwining operators and DMs}\label{sect-IO}

Condition (\ref{26}) can be clearly rewritten $\rho R=R\rho_0$: this means that $R$ is an {\em intertwining operator} (IO) between $\rho$ and $\rho_0$, and the equation is known as an {\em intertwining relation}. Of course, going back from this latter to (\ref{26}) is impossible if $R$ has no inverse. However, also in this case some interesting results can be deduced. This is what we will do in this section: we will work with non invertible intertwining operators, and see what these produce for DMs. We refer to \cite{intop} for some literature on IOs, and to \cite{bagint}-\cite{reboiro} for some results closer to what we will discuss here. Definition \ref{defRDM} is now replaced by the following (milder) alternative:

\begin{defn}\label{defGDM}
	Let $\rho_0$ be a DM, $R\in B(\Hil)$, not invertible, and $\rho\in B(\Hil)$ another bounded operator. We say that $\rho$ is a $(R,\rho_0)$-{\em generalized density matrix} (GDM) if\be
	\rho\, R=R\rho_0.
	\label{213}\en
\end{defn}

Quite often here, as we did in the previous section, we will simply call $\rho$ a GDM. The first simple remark is that $\rho\,R\in{\cal T}_1$, since $R\rho_0\in{\cal T}_1$. However, this does not imply that $\rho\in{\cal T}_1$ as well, of course, since $R^{-1}$ does not exist. Still, using (\ref{21}), (\ref{22}) and (\ref{22b}), we can deduce that, as in (\ref{29})
\begin{equation}\label{214}
	\rho\,\varphi_j=\lambda_j\varphi_j, \qquad \lambda_j\in[0,1], \quad \mbox{and}\quad \sum_j\lambda_j=1.
\end{equation}
However,  $\F_\varphi=\{\varphi_j=Re_j\}$ is no longer a Riesz basis. In fact, the following result is true, \cite{heil}:

\begin{prop}
If $R$ is surjective then $\F_\varphi$ is a frame for $\Hil$.
\end{prop} 
This follows from Corollary 8.30 of \cite{heil}, since $R\in B(\Hil)$. The one in \cite{heil} is a necessary and sufficient condition. Then, since in Definition \ref{defGDM}, $R$ is not required to be surjective, this implies that $\F_\varphi$ is not even a frame, in general. This will be clear later, in Section \ref{sectGDM}, in a concrete example.  However, formula (\ref{214}) allows us to deduce that the various $\varphi_j$ are linearly independent, at least if the eigenvalues of $\rho$ (and $\rho_0$), $\lambda_j$ are all different, which is not always the case  as we will see later. This allows us to introduce $\Lc_\varphi=l.s.\{\varphi_j\}$, the linear span of the $\varphi_j$'s, and its closure $\Hil_\varphi$. It is clear that $\Hil_\varphi\subseteq\Hil$, and $\Hil_\varphi$ is an Hilbert space\footnote{If some of the $\lambda_j$'s coincide, then we can repeat our construction of $\Hil_\varphi$ restricting only to those $\varphi_i$ which are linearly independent, i.e. only to those $\varphi_j$ for which, in (\ref{214}), the corresponding $\lambda_j$ are different. Of course, in this case, if $\dim(\Hil)<\infty$, then $\dim(\Hil_\varphi)<\dim(\Hil)$.}. $\F_\varphi$ is a basis for $\Hil_\varphi$, and it admits an unique bi-orthogonal basis $\F_\psi=\{\psi_j\}$:
\be
\langle \varphi_j, \psi_k\rangle=\delta_{j,k}, \qquad f=\sum_j \langle \varphi_j, f\rangle\,\psi_j=\sum_j \langle \psi_j, f\rangle\,\varphi_j,
\label{215}\en
for all $f\in\Hil_\varphi$. The vectors $\psi_j$ are related to $e_j$ as follows $e_j=R^\dagger\psi_j$. Indeed we have
$$
\langle e_i, e_j\rangle=\delta_{i,j}=\langle \varphi_i, \psi_j\rangle=\langle R e_i, \psi_j\rangle=\langle e_i, R^\dagger\psi_j\rangle,
$$
so that $\langle e_i, e_j-R^\dagger \psi_j\rangle=0$ for all $i$. Hence our claim follows from the completeness of $\F_e$. In this way we go back to similar results as those deduced for RDM, but restricted to $\Hil_\varphi$. In particular we can write $\rho$ (which we here identify with $\rho|_{\Hil_\varphi}$, to simplify the notation) as
\be
\rho=\sum_j\lambda_jP_j, \qquad P_j=|\varphi_j\rangle\langle\psi_j|.
\label{216}\en
If we introduce the following rank one operator $Q_j=|e_j\rangle\langle\psi_j|$ it is easy to check that
$$
P_jR=RQ_jR=RP_j^o,
$$
which reflects the same intertwining equation in (\ref{213}).  The adjoint of $\rho$ is clearly $\rho^\dagger=\sum_j\lambda_jP_j^\dagger$. Hence, at a first view, there are not many differences so far with the case of RDMs. However, this is not really so. In fact, in particular, while if $\rho$ is a RDM then $tr(\rho)=1$, if $\rho$ is a GDM  we cannot conclude that $tr(\rho)=1$ in general. This will be evident in our concrete examples below.

In view of what we have just discussed, it could be convenient to change a little bit the definition of $\rho$ in order to ensure that, even in presence of a non invertible $R$, $\rho$ has trace one. For that, we use an approach based on the idea originally discussed in \cite{bagint}. 

\begin{defn}\label{defPI}
	An operator $R\in B(\Hil)$ has the property $I$, $PI$, if $R^\dagger R$ is invertible in $B(\Hil)$.
\end{defn}
Notice that we are not requiring $R$ to be invertible. It is clear that, if $\dim(\Hil)<\infty$, the existence of $R^{-1}$ is equivalent to the $PI$, since $\det(R^\dagger R)=0$ if and only if $\det(R)=0$. The situation is different when $\dim(\Hil)=\infty$, as the following examples show.

\vspace{2mm}

{\bf Example 1:--} Let $\Hil=l^2(\mathbb{N})$. We call $S_R$ and $S_L$ respectively the right and the left shift on $\Hil$: given $a=(a_1,a_2,a_3,\ldots)\in\Hil$, we put
$$
S_Ra=(0,a_1,a_2,a_3,\ldots), \qquad S_La=(a_2,a_3,a_4,\ldots).
$$
It is clear that $S_LS_Ra=a$, but $S_RS_La\neq a$. Hence $S_L$ is not the inverse of $S_R$. Now, if we put $R=S_R$, it follows that $R^\dagger=S_L$ and $R^\dagger R=S_LS_R=\1$, which is clearly invertible. However, as we have seen, $S_R^{-1}$ does not exist. Hence $S_R$ has the $PI$.

\vspace{2mm}

{\bf Example 2:--} Let $\Hil$ be a generic (infinite-dimensional) Hilbert space and $\F_e=\{e_j\}$ an ONB for $\Hil$. Let us further consider an increasing bounded sequence $\{\epsilon_n\}$ such that $0=\epsilon_0<\epsilon_1<\epsilon_2<\cdots\leq\epsilon_\infty<\infty$. We introduce $R=\sum_{n=0}^{\infty}\epsilon_{n+1}|e_{n+1}\rangle\langle e_n|$. This is a densely defined operator with domain $D(R)\supseteq\Lc_e$, the linear span of the $e_n$'s, with $Re_k=\epsilon_{k+1}e_{k+1}$, $k\geq0$. The adjoint of $R$ turns out to satisfy the lowering condition $R^\dagger e_0=0$ and $R^\dagger e_k=\epsilon_k\,e_{k-1}$, $k\geq1$. We can easily find that
$$
R^\dagger R=\sum_{n=0}^{\infty}\epsilon_{n+1}^2|e_{n}\rangle\langle e_n|, \qquad \mbox{while}\qquad RR^\dagger=\sum_{n=0}^{\infty}\epsilon_{n+1}^2|e_{n+1}\rangle\langle e_{n+1}|.
$$
We observe that $R\in B(\Hil)$, with $\|R\|\leq\epsilon_\infty$, and $R$ is not invertible. However $R^\dagger R$ admits inverse, $(R^\dagger R)^{-1}=\sum_{n=0}^{\infty}\epsilon_{n+1}^{-2}|e_{n}\rangle\langle e_n|$, while $(RR^\dagger)^{-1}$ does not exist, since $0\neq e_0\in \ker(RR^\dagger)$, so that $RR^\dagger$ is not injective.

\vspace{2mm}

Summarizing, these examples (together with those in \cite{bagint}) show that property $PI$ is not trivial, and it makes sense to consider it in our context. In fact, in this case, we  can identify the set $\F_\psi$ above: if we put $\psi_j=R(R^\dagger R)^{-1}e_j$, it is clear that 
$$
\langle \varphi_i, \psi_j\rangle=\langle Re_i, R(R^\dagger R)^{-1}e_j\rangle=\langle e_i, R^\dagger R(R^\dagger R)^{-1}e_j\rangle=\langle e_i, e_j\rangle=\delta_{i,j}.
$$
Furthermore we can check that, using the fact that $R$ and $(R^\dagger R)^{-1}$ are continuous,
$$
\sum_j\langle \varphi_j, f\rangle\psi_j=\sum_j\langle \psi_j, f\rangle\varphi_j=R(R^\dagger R)^{-1} R^\dagger f,
$$
for all $f\in\Hil_\varphi$. Therefore $\hat f=\sum_j\langle \varphi_j, f\rangle\psi_j-f$ and $\check f=\sum_j\langle \psi_j, f\rangle\varphi_j-f$ both belong to the $\ker(R^\dagger)$. Notice now that, in particular,
$$
0=\langle R^\dagger \check f, e_j\rangle=\langle  \check f, R\,e_j\rangle=\langle  \check f, \varphi_j\rangle,
$$
which implies that $\check f=0$, due to the fact that $\F_\varphi$ is total in $\Hil_\varphi$, and $\check f\in\Hil_\varphi$. Hence $f=\sum_j\langle \psi_j, f\rangle\varphi_j$. The fact that $\hat f=0$ can be proved similarly, at least if  $R(R^\dagger R)^{-1}e_j\in\Hil_\varphi$ for all $j$. Then we conclude that, under our assumptions,  $\F_\varphi$ and $\F_\psi$ are bi-orthonormal bases in $\Hil_\varphi$. It is now simple to deduce that
\be
\rho=R\rho_0(R^\dagger R)^{-1}R^\dagger.
\label{217}\en
The first obvious remark is that this formula extends the one in (\ref{212}), which is recovered if $R^{-1}$ exists. Moreover we have $tr(\rho)=tr(\rho_0)=1$, using the property $tr(AB)=tr(BA)$ of the trace. It is also easy to understand that, with our special choice of $\psi_j$, we still have $\rho=\sum_j\lambda_jP_j$, and  $\rho^2=\sum_j\lambda_j^2P_j$. 

We can use $\rho$ as in the previous sections to define a linear functional as in (\ref{28}), but with some changes. In this case we have
\be
\Phi_{\rho}(X)=tr(\rho X)=tr(R\rho_0(R^\dagger R)^{-1}R^\dagger X)=\Phi_{\rho_0}(\tilde X_R),
\label{218}\en
where $\tilde X_R=(R^\dagger R)^{-1}R^\dagger XR$. $\Phi_\rho$ is not positive, and it is not true that, given any $X=X^\dagger$, then $\Phi_\rho(X)\in\mathbb{R}$. In fact, this was not true even in the simpler case of RDMs. On the other hand, $\Phi_\rho$ is linear, normalized, and continous: if $X_n\rightarrow X$ in $B(\Hil)$, then $\Phi_\rho(X_n)\rightarrow\Phi_\rho(X)$ in $\mathbb{C}$. This is a consequence of the inequality
$$
|\Phi_\rho(X)|\leq\|(R^\dagger R)^{-1}\|\|R\|^2\|X\|,
$$
$\forall X\in B(\Hil)$. Going back to the (lack of) positivity of $\Phi_\rho$, we can check that, if $X>0$ is such that $[R^\dagger R,X]=0$, then $\Phi_\rho(X)>0$.

We conclude this abstract analysis of DMs introducing the notion of pure states also for GDM. 

\begin{defn}
	The GDM in (\ref{213}) is a {\em generalized pure state} (GPS) if $\rho_0$ is a pure state, i.e. if it exists a normalized vector $\Phi_0$ such that $\rho_0=|\Phi_0\rangle\langle\Phi_0|$. In this case, calling $\varphi_0=R\Phi_0$ and $\psi_0=R(R^\dagger R)^{-1}\Phi_0$ we can write
	\begin{equation}\label{219}
		\rho=|\varphi_0\rangle\langle\psi_0|.
	\end{equation}
\end{defn}
Connected to this we can introduce the following operator, which we call {\em generalized entropy operator} (GEO): $S(\rho)$ is a GEO if the following intertwining relation holds:
\be
S(\rho)R=RS(\rho_0).
\label{220}\en
The counterpart of Theorem \ref{thm2} is the following:

\begin{thm}\label{thm3}
	A GDM $\rho$ is a GPS if and only if one of the following equivalent properties is satisfied:
	
	{\bf p1$''$.} $tr(\rho^2)=1$.
	
	{\bf p2$''$.} $S(\rho)=0$ on $\Hil_\varphi$.

\end{thm}

The proof is similar to the one of Theorem \ref{thm2} and will not be repeated.

In the following sections we will see how our results look like in two concrete examples.

\section{Examples of generalized DMs}\label{sect-ex}

In this section, we present different examples in which a $(R,\rho_0)$-Riesz density matrix (RDM) can naturally be defined as a suitable deformation of a density operator through \eqref{26}. The first application relies on a construction of a RDM starting from a deformation related to a classical gain and loss system described by a non-Hermitian Hamiltonian. Starting from a DM dependent on time and applying a similarity deformation $R$, we obtain a RDM that preserve trace, entropy and purity (i.e. $tr(\rho^2)$).  The other applications  are connected to a finite-dimensional version of the Swanson oscillator, once again described by a non-Hermitian Hamiltonian \cite{swan,reboiro}. In these case we construct a RDM starting from the possibility of moving around  \textit{exceptional points} and analyze whether such situation induces critical behaviors like the totally loss of purity.
We also determine the conditions for defining a GDM when the deformation matrix $R$ is no more invertible, as described in Section \ref{sect-IO}. 
In all the examples, we shall discuss the conditions under which the RDMs defines a pure state (RPS) or a fully mixed state.
We emphasize that our primary objective in this section is to validate our mathematical framework by deriving the RDM through appropriate deformations of some DM connected to some models somehow related to pseudo-Hermitian quantum mechanics, keeping in mind that, however, there are numerous ways to deform a DM and induce a RDM (or a GDM).

\subsection{Application I: a two-state non-Hermitian system}\label{sect-2D}

In this section we will consider a non-Hermitian system, in particular an open two-state system with balanced gain and loss terms, in the regime of spontaneously broken $\mathcal{PT}$ symmetry, as analyzed in \cite{Felski}.

We begin introducing the two-state Hamiltonian:
\be \label{H-2}
H = \begin{pmatrix}
	r e^{i \theta} & d \\
	d & r e^{-i \theta}  
\end{pmatrix},
\en
were $r$ and $d \in \mathbb{R}$.  Notice that $H \neq H^\dagger$, if $\theta \neq k \pi$, $k \in \mathbb{Z}$.

First of all we determine eigenvalues
and eigenvectors of the system (observing the presence of exceptional points), and then we will analize a DM related to the Hamiltonian and its entropy, defined as in Section \ref{sect-sim}.

The eigenvalues of $H$  are:
\be \label{mu-2}
\mu_{\pm} = r \cos(\theta) \pm \sqrt{d^2 - r^2 \sin^2(\theta)}, 
\en

and their correspondent eigenvectors are:

\be \label{phi-2}
\varphi_{\pm} =(\overline{A_\pm})^{-1} \begin{pmatrix}
	i r \sin(\theta) \pm \sqrt{d^2 - r^2 \sin^2(\theta)}\\
	d
\end{pmatrix},
\en
where $A_\pm = \sqrt{2d^2 -2r^2 \sin^2(\theta) \pm 2 i r \sin(\theta)\sqrt{d^2 - r^2 \sin^2(\theta)}}$, are normalization factors, whose usefulness will be explained immediately afterwards.

It is clear that eigenvalues and eigenvectors depend strongly on the values of the parameters $d$, $r$ and $\theta$, and exceptional points arise when $d^2 = r^2 \sin^2(\theta)$, so that eigenvalues and eigenvectors coalesce. 
Furthermore, when $d^2 > r^2 \sin^2(\theta)$ the eigenvalues are reals and the system is in \textit{unbroken region}, otherwise, they will be complex and the system is in the \textit{broken region}. 

In the {\it unbroken region}, when $d^2 > r^2 \sin^2(\theta)$, $\mu_{\pm}$ are eigenvalues also for $H^\dagger$, and its correspondent eigenvectors are:
\be \label{psi-2}
\psi_{\pm} =(A_\pm)^{-1} \begin{pmatrix}
	-i r \sin(\theta) \pm \sqrt{ d^2 - r^2 \sin^2(\theta)}\\
	d
\end{pmatrix}.
\en

With this choice of $A_\pm$, we have normalized the eigenvectors in order to have $\langle \varphi_j, \psi_i \rangle =  \delta_{j,i}$, and the families $\mathcal{F}_{\varphi} = \{\varphi_\pm\}$ and $\mathcal{F}_{\psi} = \{\psi_\pm\}$ are Riesz-basis, since the model is defined on a finite dimensional Hilbert space. 

In the other case, when $d^2 < r^2 \sin^2(\theta)$, i.e. in the {\it broken region}, we have:
\begin{equation}
	\mu_{\pm} = r \cos(\theta) \pm i \sqrt{ r^2 \sin^2(\theta)-d^2},
\end{equation}
and the eigenvalues and  eigenvectors of $H^\dagger$  are the following ones:
\begin{equation}
	\nu_{\pm} =  \overline{\mu_{\pm}},  \qquad  \tilde{\psi}_{\pm} =  \psi_{\mp}.
\end{equation}
In this case, the normalization factors  $A_\pm$ became real quantities , being $\sqrt{d^2- r^2 \sin^2(\theta)}= i \sqrt{ r^2 \sin^2(\theta)-d^2}$, and  the families $\mathcal{F}_{\varphi}$ and $\mathcal{F}_{\tilde{\psi}}= \{\tilde\psi_\pm\} $ are also bi-orthogonal Riesz-basis, because $\langle \varphi_j, \tilde{\psi}_i \rangle =  \delta_{j,i}$.

\subsubsection{Density matrices}

Our main interest is to show an example of a $(R,\rho_0)$-Riesz density matrix, as defined in Section \ref{sect-sim}. Hence we start with a generic density matrix $\rho_0(0)=\begin{pmatrix} 
	c_1 &  c_2\\
	c_3 &  c_4
\end{pmatrix}$, in which $c_1 + c_4 =1$, $c_3 = c_2^*$, and the $c_j$'s are chosen in such a way $\rho_0(0)$ is positive, then we consider the usual Von Neumann evolution equation starting from an Hermitian Hamiltonian, i.e.
\begin{equation}\label{ro-t-eq}
	\dfrac{d}{dt}\rho_0(t)= -i [H_0,\rho_0(t)],
\end{equation}
where we have put $\hbar=1$ and $H_0=H(\theta=0)$. So, we obtain a density matrix depending on time, $\rho_0(t) =$
\begin{equation}
	\dfrac{1}{2}\begin{pmatrix} 
		1+ i (c_2-c_3) \sin\Omega+(c_1-c_4) \cos \Omega &  c_2+c_3 + i (c_1-c_4) \sin \Omega - (c_3-c_2) \cos\Omega\\
		c_2+c_3 -i (c_1-c_4) \sin\Omega+(c_3-c_2)
		\cos\Omega&  1 -i (c_2-c_3) \sin\Omega-(c_1-c_4) \cos\Omega
	\end{pmatrix},
\end{equation}
where $\Omega=2 d t$. 

To construct an RDM, we consider a particular $\rho_0(t)$, with $c_1=\frac{2}{3}$, $c_2=c_3=0$, $c_4=\frac{1}{3}$. With this choice, $\rho_0(0)$ is diagonal and positive, other than Hermitian, and we have 
\begin{equation}\label{ro-t}
	\rho_0(t) =  \dfrac{1}{2}\begin{pmatrix} 
		1+ \frac{1}{3} \cos (2dt) &   \frac{1}{3} i \sin (2 d t) \\
		- \frac{1}{3} i \sin (2d t)&  1 -\frac{1}{3} \cos (2dt)
	\end{pmatrix}
\end{equation}

As in Section \ref{sect-sim}, we can obtain a RDM using a bounded invertible operator $R$ with bounded inverse. A particular example of $R$ can be constructed by using the eigenstates $\varphi_i$  in \eqref{phi-2}, i.e. 

\begin{equation}\label{R-2}
	R = \begin{pmatrix} 
		\dfrac{2iy + \sqrt{1-4y^2} }{ \sqrt{2(1-4y^2)-4 i y \sqrt{1-4y^2}}} &  \dfrac{2iy - \sqrt{1-4y^2} }{ \sqrt{2(1-4y^2)+4 i y \sqrt{1-4y^2}}}\\
		\dfrac{1}{\sqrt{2(1-4y^2)-4 i y \sqrt{1-4y^2}}}&  \dfrac{1}{  \sqrt{2(1-4y^2) + 4 i y \sqrt{1-4y^2}}}
	\end{pmatrix},
\end{equation}
where we have fixed $d=0.5$, $r=1$ and we have introduced $y=\sin(\theta)$, to simplify the notation. Using this matrix to deform $\rho_0(t)$ we obtain a RDM $\rho_\theta(t)=R \rho_0(t)R^{-1}$, that is:
\begin{equation}\label{rho-theta}
	\rho_\theta(t)=\begin{pmatrix} 
		\dfrac{1}{2}+ \dfrac{1}{3}  \dfrac{ i y  \cos(t)}{\sqrt{1-4y^2}} + \dfrac{2}{3} y \sin(t) &  
		\dfrac{\cos(t)}{6\sqrt{1-4y^2}} + \dfrac{i (1-16y^2) \sin(t)}{6}
		\\
		\dfrac{\cos(t)}{6\sqrt{1-4y^2}} - \dfrac{i \sin(t)}{6} & 
		\dfrac{1}{2}  - \dfrac{1}{3}\dfrac{i y \cos(t)}{\sqrt{1-4y^2}} - \dfrac{2}{3}  y \sin(t) 
	\end{pmatrix},
\end{equation}

Although this matrix depends on time ( and on the deformation parameter  $\theta$ through $y$), its trace is preserved and it is always equal to 1, as expected. Furthermore are preserved its purity and  entropy, that are equal respectively to $\frac{5}{9}$ and $\log(3)-\frac{2}{3}\log(2)$, which are the same values we can obtain from  $\rho_0(t)$. Therefore we are not in presence of a RPS, since the purity is never equal to 1, nor entropy equal to 0. This situation is not surprising because we are deforming the DM with a similarity deformation, that preserve the trace (also in the computation of the entropy and of the purity), and since our $\rho_0(0)$ is not a pure state, and $\rho_0(t)$ in (\ref{ro-t}) is not a pure state either.

\subsection{Application II: The finite dimensional Swanson model}\label{sect-Swa}
Let's now introduce the following finite dimensional version of the Swanson Hamiltonian:
\be \label{H}
H = c^\dagger c + \alpha_1 c^2 + \alpha_2 (c^{\dagger})^2,
\en
where $c$ is a lowering operator, satisfying the (truncated) CCR
$
[c,c^\dag]=(\1-3P_j^o)
$
i.e. $c e_{i+1} = \sqrt{i} e_i$ for $i= 1,2$, $c^\dagger e_{i} = \sqrt{i} e_{i+1}$ for $i= 2,3$ and with $c e_{1}=(c^\dag) e_{3}=0$, where $e_j$ are the canonical o.n. vectors of the $\mathbb{R}^3$ basis, and $\alpha_1$ and $\alpha_2$ are real numbers. A matrix realization of $c$ and $H$ is the following:
\be \label{H-matrix}
c=\begin{pmatrix}
0 & 1 & 0 \\
0 & 0 & \sqrt{2} \\
0 & 0 & 0
\end{pmatrix}
\qquad \text{and} \qquad
H = \begin{pmatrix}
0 & 0 & \sqrt{2} \alpha_1 \\
0 & 1 & 0 \\
\sqrt{2} \alpha_2 & 0 & 2
\end{pmatrix},
\en
This Hamiltonian is clearly non-Hermitian ($H^\dagger \neq H$), when $\alpha_1 \neq \alpha_2$. The eigenvalues of $H$  are:
\be \label{mu}
\mu_1 = 1, \qquad  \mu_2 = 1 - \sqrt{1 + 2 \alpha_1 \alpha_2} \quad \text{and} \quad \mu_3 = 1 + \sqrt{1 + 2 \alpha_1 \alpha_2},
\en
and their correspondent eigenvectors are:
\be \label{phi}
\varphi_1 = \begin{pmatrix}
0\\
1\\
0
\end{pmatrix},
\qquad 
\varphi_2 = \begin{pmatrix}
-\dfrac{ h_3 \mu_3}{\sqrt{2} \alpha_2 } \\
0\\
h_3
\end{pmatrix}, \quad \text{and} \quad 
\varphi_3 = \begin{pmatrix}
-\dfrac{ h_2 \mu_2 }{\sqrt{2} \alpha_2}\\
0\\
h_2
\end{pmatrix}.
\en 
If $1+2\alpha_1\alpha_2\geq0$, $\mu_1,\mu_2,\mu_3$ are also eigenvalues of $H^\dagger$,  with eigenvectors:
\be \label{psi}
\psi_1 = \begin{pmatrix}
0\\
1\\
0
\end{pmatrix},
\qquad 
\psi_2 = \begin{pmatrix}
-\dfrac{ h_3 \mu_3}{ \sqrt{2} \alpha_1} \\
0\\
h_3
\end{pmatrix}, \quad \text{and} \quad 
\psi_3 = \begin{pmatrix}
-\dfrac{h_2 \mu_2}{ \sqrt{2} \alpha_1} \\
0\\
h_2
\end{pmatrix},
\en
where we have defined $h_2= \left( \dfrac{\mu_2^2}{2 \alpha_1 \alpha_2} + 1 \right)^{-1/2}$ and $h_3= \left( \dfrac{\mu_3^2}{2 \alpha_1 \alpha_2} + 1 \right)^{-1/2}$.
The eigenvectors are bi-normalized: $\langle \varphi_j, \psi_i \rangle = \delta_{j,i}$, and the families $\mathcal{F}_{\varphi}$ and $\mathcal{F}_{\psi}$ form two Riesz basis when $1 + 2 \alpha_1 \alpha_2 \neq 0$: we will stress this point later when introducing the matrix $R$ in \eqref{Rm}. When $1 + 2 \alpha_1 \alpha_2 = 0$  we observe that $\mu_2$ and $\mu_3$ coalesce, along with their corresponding eigenvectors, and hence $1 + 2 \alpha_1 \alpha_2 = 0$ describes a curve (i.e., a hyperbola) of exceptional points. In this situation, the family $\mathcal{F}_{\varphi}$ does not form a Riesz basis.
This situation is a typical characterization of the formation of an \textit{exceptional point}, which marks the transition from the unbroken to the broken region. In particular, when $1 + 2 \alpha_1 \alpha_2 > 0$, we are in the \textit{unbroken region}, and the eigenvalues are real. Conversely, when $1 + 2 \alpha_1 \alpha_2 < 0$, we have a pair of complex conjugate eigenvalues $\mu_2 = \overline{\mu_3}$, indicating the \textit{broken region}.
In the latter case, to recover the bi-orthogonality of the eigenvectors of $H$ and $H^\dagger$, we reorder the eigenvectors from the vectors $\psi$:
\begin{equation} \label{psi-tilde} 
	\tilde{\psi}_2 = \psi_3 \qquad \text{and} \qquad \tilde{\psi}_3 = \psi_2,
\end{equation}
so that we again have $\langle \varphi_j, \tilde{\psi_i} \rangle = \delta_{j,i}$. We are using here the same notation already adopted for the previous example.

\subsubsection{RDM I: a time dependent case}
Let us consider	the following 3-dimensional  time dependent $\rho$ whose entries are:

\begin{eqnarray*}
	\rho_{11} & = & \frac{-X \mu_3 \lambda_2 + \mu_2 (\lambda_3 + \alpha_1^2 (\lambda_1 + \lambda_3)) + \alpha_1^2 \mu_2 (-\lambda_1 + \lambda_3) C_{h,2}}{X (\mu_2 - \mu_3)} \\
	\rho_{12} & = & -\frac{h_2 \alpha_1 \mu_2 (\lambda_1 - \lambda_3) (i\sqrt{X} - i\sqrt{X} C_{h,2} - X S_{h})}{2 (X)^{3/2} \alpha_2} \\
	\rho_{13} & = & \frac{\mu_2 \mu_3 (-\lambda_2 + \lambda_3 + \alpha_1^2 (\lambda_1 - 2 \lambda_2 + \lambda_3) + \alpha_1^2 (-\lambda_1 + \lambda_3) C_{h,2})}{\sqrt{2} X \alpha_2 (\mu_2 - \mu_3)} \\
	\rho_{21} & = & \frac{2 \alpha_1 \alpha_2 S_h\, (i X C_h -\sqrt{X} S_h) (\lambda_1 - \lambda_3)}{h_2 (X)^{3/2} (\mu_2 - \mu_3)} \\
	\rho_{22} & = & \frac{\lambda_1 + \alpha_1^2 (\lambda_1 + \lambda_3) + \alpha_1^2 (\lambda_1 - \lambda_3) C_{h,2}}{X} \\
	\rho_{23} & = & \frac{\sqrt{2} \alpha_1 \mu_3 S_h\, (i X C_h -\sqrt{X} S_h) (\lambda_1 - \lambda_3)}{h_2 (X)^{3/2} (\mu_2 - \mu_3)} \\
	\rho_{31} & = & \frac{\sqrt{2} \alpha_2 (\lambda_2 - \lambda_3 - \alpha_1^2 (\lambda_1 - 2 \lambda_2 + \lambda_3) + \alpha_1^2 (\lambda_1 - \lambda_3) C_{h,2})}{X (\mu_2 - \mu_3)} \\
	\rho_{32} & = & -\frac{\sqrt{2} h_2 \alpha_1 S_h (-i X C_h -\sqrt{X} S_h) (\lambda_1 - \lambda_3)}{(X)^{3/2}} \\
	\rho_{33} & = & \frac{(\mu_2 + 2 \alpha_1^2 \mu_2) \lambda_2 - \mu_3 (\lambda_3 + \alpha_1^2 (\lambda_1 + \lambda_3)) + \alpha_1^2 \mu_3 (\lambda_1 - \lambda_3) C_{h,2}}{X (\mu_2 - \mu_3)}
\end{eqnarray*}
	where we have defined $X=1 + 2 \alpha_1^2,S_h=\sin(\sqrt{1 +2\alpha_1^2}\,t),C_h=\cos(\sqrt{1 +2\alpha_1^2}\,t),C_{h,2}=\cos(2\,\sqrt{1 +2\alpha_1^2}\,t)$, $\mu_1,\mu_2,\mu_3$ are the eigenvalues of the finite dimensional Swanson model,  $\lambda_1,\lambda_2,\lambda_3$ are chosen  to satisfy $\sum_j\lambda_j=1$,  and where the $\lambda_j$ coefficients will be defined shortly. Clearly $\rho$ is well defined whenever $\mu_2\neq \mu_3$, that is when $\alpha_1\alpha_2\neq -1/2$.
It is possible to check that $\rho$ is actually a $RDM$ related to a $DM$ trough 	\eqref{26} where the $\rho_0(t)$ is defined as 	\begin{equation}
	\rho_0(t)=\begin{pmatrix}
		\frac{\lambda_1 + C_{h,2} \alpha_1^2 (\lambda_1 - \lambda_3) + \alpha_1^2 (\lambda_1 + \lambda_3)}{X} & 0 & -\frac{\sqrt{2} \, S_h (i\,S_h \sqrt{X} + i \, C_h \, X) \alpha_1 (\lambda_1 - \lambda_3)}{(X)^{3/2}} \\
		0 & \lambda_2 & 0 \\
		-\frac{\sqrt{2} \, S_h (i\,S_h \sqrt{X} - i \, C_h \, X) \alpha_1 (\lambda_1 - \lambda_3)}{(X)^{3/2}} & 0 & \frac{\lambda_3 + C_{h,2} \alpha_1^2 (-\lambda_1 + \lambda_3) + \alpha_1^2 (\lambda_1 + \lambda_3)}{X}
	\end{pmatrix}
\end{equation}
	and where $R$ is the matrix consisting of the eigenvectors of the family $\mathcal{F}_{\varphi} $, i.e.
	\be \label{Rm}
	R = \begin{pmatrix}
		0 & -\dfrac{ h_3 \mu_3}{\sqrt{2} \alpha_2}  & -\dfrac{h_2 \mu_2}{\sqrt{2} \alpha_2}  \\
		1 & 0 & 0 \\
		0 & h_3 & h_2
	\end{pmatrix}.
	\en
To clarify our choices here we observe that $\rho_0(t)$ is the evolved density matrix obtained via the usual von Neumann evolution when the Hermiticity of the Swanson Hamiltonian is restored, that is, when $\alpha_2 = \alpha_1$. For simplicity, we consider the initial condition $\rho_0(0) = \sum_j \lambda_j |e_j\rangle \langle e_j|$, and if at least two of the $\lambda_j$ are different from zero, this initial condition represents an ensemble of states: $\rho_0(0)$ is not pure. In other words: we start from $\rho_0(0)$ and let it evolve using, as in the previous example, equation (\ref{ro-t-eq}) to deduce $\rho_0(t)$. In this case, $H_0$ is the Hamiltonian in (\ref{H}) with $\alpha_1=\alpha_2$. Then we use the operator $R$ to deform $\rho_0(t)$ as in (\ref{26}), and we recover a very complicated matrix $\rho$, whose entries are given above. This is our RDM.
 Notice that, by construction, $R$ is not unitary and it is not invertible at the exceptional point, that is when $\alpha_1\alpha_2 = -1/2$ or $\mu_2 = \mu_3$. Conversely, when $\alpha_1\alpha_2 \neq -1/2$, $R$ is invertible, and the vectors of the families $\mathcal{F}_{\varphi}$ and $\mathcal{F}_{\psi}$ can be obtained from the canonical basis $\{e_j\}$ in the following way: $\varphi_j = R e_j$ and $\psi_j = (R^{-1})^\dagger e_j$, and they satisfy \eqref{29} and \eqref{29-bis}.

\subsubsection{Entropy and purity}
	It is clear that, at least for $\alpha_1 \alpha_2 \neq -1/2$, $\rho(t)$ and $\rho_0(t)$ share the same trace, as well as their derived quantities such as entropy (in view of \eqref{211}) and purity. To highlight possible critical behaviors, we define the initial conditions on the $\lambda_j$'s related to the parameters $\alpha_1$ and $\alpha_2$. Specifically, we set
	$$
	\lambda_j = \frac{|\mu_j|^2}{\sum_j |\mu_j|^2}
	$$
	which guarantees that $\rho_0(0)$ is always positive definite with unit trace, independently of the values of $\alpha_1$ and $\alpha_2$. When $\alpha_1 \alpha_2 \rightarrow \pm \infty$, we have  $\lambda_1 \rightarrow 0$ and $\lambda_{2,3} \rightarrow 1/2$.
	
	Due to the Hermitian evolution of $\rho_0(t)$, and because the Hamiltonian is time-independent, the entropy and the purity of $\rho_0(t)$ are preserved in time, as well as those of $\rho(t)$.
	The behaviors of the purity $tr(\rho^2(t))$ and the trace of the entropy operator $S(\rho(t))$ are shown in Figures \ref{3da}-\ref{3dc} by varying $\alpha_2$ while keeping $\alpha_1 = 1$. 
	As $\alpha_2$ approaches the exceptional point, $\alpha_2 \rightarrow -1/2^{\pm}$, the purity tends to its minimum value of $1/3$, and the entropy reaches its maximum value of $\log(3)$ (Figure \ref{3da}). This indicates that the  RDM describes a complete mixture of states near this point. 
	For both decreasing and increasing values of $\alpha_2$, Figures \ref{3db}-\ref{3dc}, due to the asymptotic behavior of the $\lambda_j$'s, the purity and entropy attain the asymptotic values $\operatorname{tr}(\rho^2)_{\alpha_2 \rightarrow \infty} = 1/2$ and $S(\rho)_{\alpha_2 \rightarrow \infty} = \log(2)$.
	
	\begin{figure}[!ht]
	\hspace*{-0.75cm}\subfigure[]{\includegraphics[width=6.25cm]{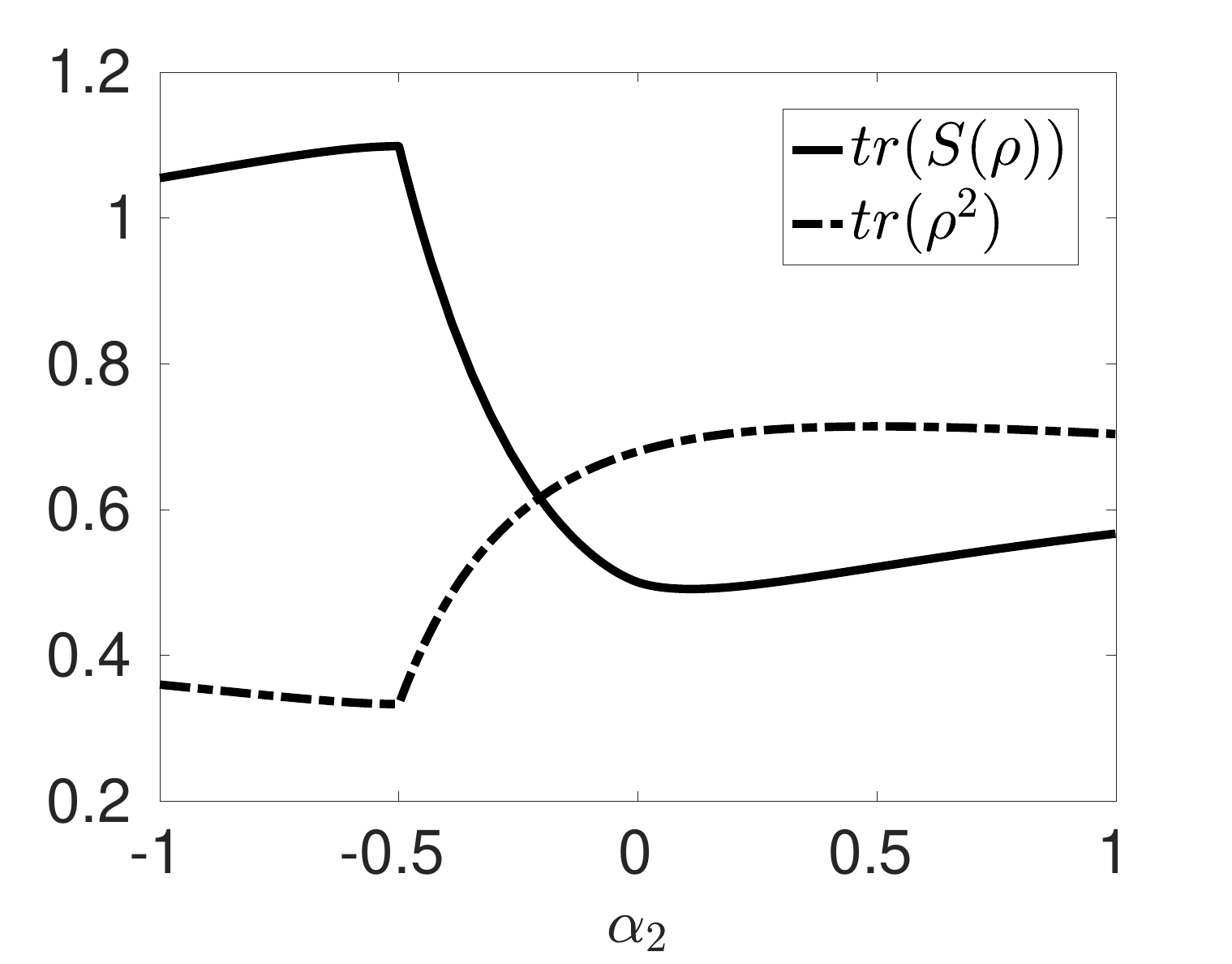}\label{3da}}
	\hspace*{-0.68cm}\subfigure[]{\includegraphics[width=6.25cm]{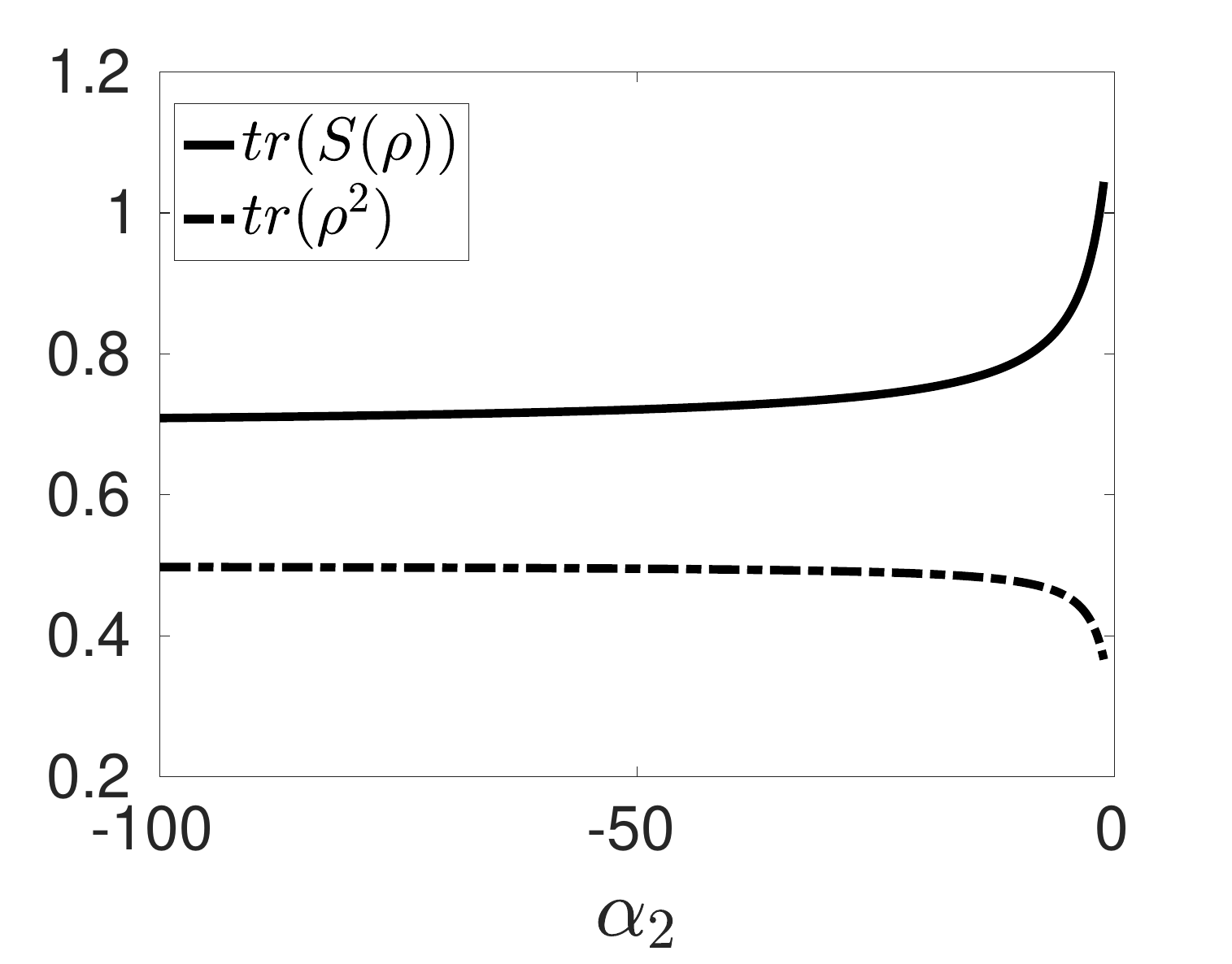}\label{3db}}
	\hspace*{-0.68cm}\subfigure[]{\includegraphics[width=6.25cm]{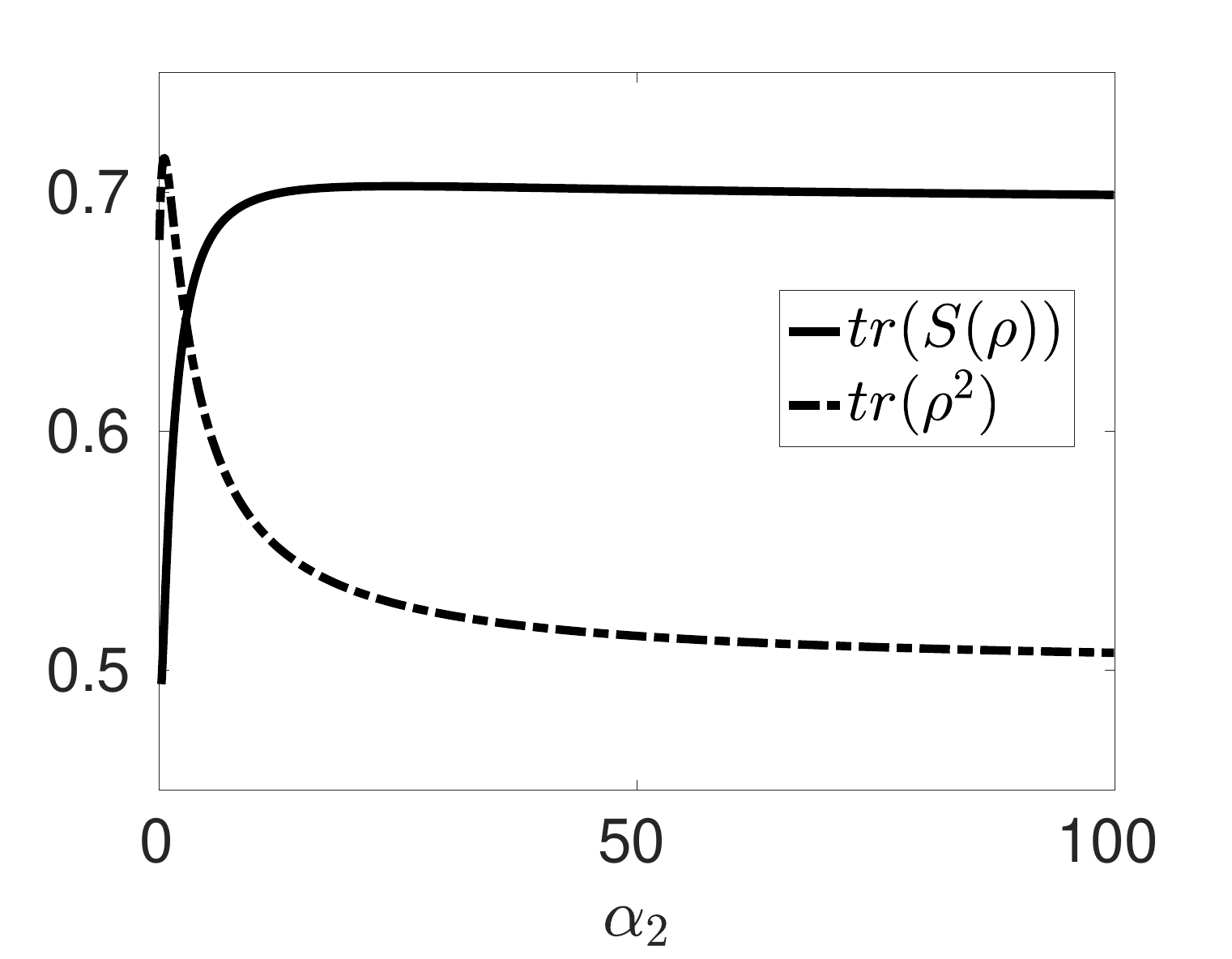}\label{3dc}}
	\caption{\textbf{(a) }Behavior of purity $tr(\rho^2(t))$ and entropy $tr(S(\rho(t)))$ in the vicinity of the exceptional point $\alpha_2 = -1/2$, with $\alpha_1 = 1$ for the RDM I example. At the exceptional point, $R$ is not invertible, and the purity tends to its minimal value of 1/3, while the entropy to its maximal value of $\log(3)$. \textbf{(b)} Behavior of purity  and entropy for decreasing values of $\alpha_2$ with $\alpha_1 = 1$. As $\alpha_2$ moves away from the exceptional point, the purity and entropy approach their asymptotic values of $1/2$ and $\log(2)$, respectively. \textbf{(c)} Same as \textbf{(b)} but for increasing $\alpha_2$. }
\end{figure}
	
	\subsubsection{RDM II: a time independent case}
	
	As done in the previous section we want to recover a RDM starting from a three dimensional density matrix  describing a physical system. Consider
	$$
	\rho=\begin{pmatrix}
		\frac{1}{2} \left(\frac{\lambda_2}{\sqrt{2 \alpha_1 \alpha_2 + 1}} - \frac{\lambda_3}{\sqrt{2 \alpha_1 \alpha_2 + 1}} + \lambda_2 + \lambda_3\right) & 0 & \frac{\alpha_1 (\lambda_3 - \lambda_2)}{\sqrt{4 \alpha_1 \alpha_2 + 2}} \\
		0 & \lambda_1 & 0 \\
		\frac{\alpha_2 (\lambda_3 - \lambda_2)}{\sqrt{4 \alpha_1 \alpha_2 + 2}} & 0 & \frac{1}{2} \left(-\frac{\lambda_2}{\sqrt{2 \alpha_1 \alpha_2 + 1}} + \frac{\lambda_3}{\sqrt{2 \alpha_1 \alpha_2 + 1}} + \lambda_2 + \lambda_3\right) \\
	\end{pmatrix}.
	$$
	This is an RDM   well defined whenever $\alpha_1\alpha_2>-1/2$, and one can verify that it  can formally written  as $\rho=R \rho_0 R^{-1}$, where $R$ is again given in \eqref{Rm}. Here $\rho_0$ is a DM describing a system which is in equilibrium due to an immersion in a heath bath, \cite{brody},
	and whose expression is
	$\rho_0 = \sum_j \lambda_j |e_j\rangle \langle e_j|$, where 
	\be \label{lambda}
	\lambda_j = \dfrac{e^{- \beta \mu_j}}{\sum_je^{- \beta \mu_j}},
	\en
	being $ \beta = 1/ k T$ with $k$ the Boltzmann's constant,  $T$  the temperature of the bath, and where we are using $\mu_1,\mu_2,\mu_3$,  the eigenvalues of the Swanson's model. We stress here that $\rho_0$ is a DM only in the case the $\mu_j's$ are real, that is for $\alpha_1\alpha_2\geq-1/2$, since otherwise the constraint $\lambda_j\in[0,1]$ would be violated, so that we shall work only in the \textit{un-broken region} of the Swanson's model.  
The eigenvalues of $\rho_0$ are all equal to 1/3 when $\alpha_1\alpha_2=-1/2$, and reach asymptotic values $\lambda_2\rightarrow1,\lambda_{1,3}\rightarrow0$ as $\alpha_1\alpha_2\rightarrow+\infty$. This means that the system, asymptotically, is in the pure state  $|e_1\rangle$ with a rate that increases with  $\beta$.
% {\red (QUESTA FIGURA LA LEVEREI, BASTA DESCRIVERE a parole IL COMPORTAMENTO DEI LAMBDA)}. 
%
%\begin{center}
%	\begin{figure}[!ht]
%	\centering
%	\includegraphics[width=13cm]{Eigenvalues-3}
%	\caption{The eigenvalues of $\rho_0$ in the RDM II case for $\beta =0.5$ versus $\alpha_1\alpha_2$. }\label{fig-eigenvalues}
%\end{figure}
%\end{center}	
%	
%	
In this configuration, we can formally derive the entropy operator for $\rho$, 
$$S(\rho) = - \sum_j \lambda_j \log(\lambda_j) |\varphi_j\rangle \langle \psi_j|,$$
in accordance with \eqref{25} and \eqref{211}, with entropy given by:
$$
	tr(S(\rho)) = - \frac{\log\left(\frac{1}{1+e^x+e^{2X}}\right) + e^X \log\left(\frac{1}{1+2\cosh(X)}\right) + e^{2X} \log\left(\frac{e^X}{1+2\cosh(X)}\right)}{1+e^X+e^{2X}},
$$
and the purity 
$$tr(\rho^2)=1-\frac{2}{2 \cosh \left( X \right)+1},$$
where \(X = \beta \sqrt{1 + 2 \alpha_1 \alpha_2}\). The behaviors of the entropy and the purity are depicted in Figure \ref{fig_unbroken} for various values of $\beta$ and under the condition $\alpha_1\alpha_2 > -1/2$.
We observe  that when \(\alpha_1\alpha_2 \rightarrow -1/2^+\), close to the formation of the exceptional point, the entropy $tr(S(\rho))$ reaches its maximum allowed value of $\log(3)$, while the purity tends to 1/3,  indicating a fully mixed state. Instead, in the asymptotic regime $\alpha_1\alpha_2\rightarrow\infty$ we obtain \(tr(S(\rho))\rightarrow0\) and \(tr(\rho^2)\rightarrow1\) meaning that, asymptotically, the RDM become a RPS represented by $|\varphi_2\rangle \langle \psi_2|$. 
	\begin{figure}[!ht]
	\centering
	\subfigure[]{\includegraphics[width=8cm]{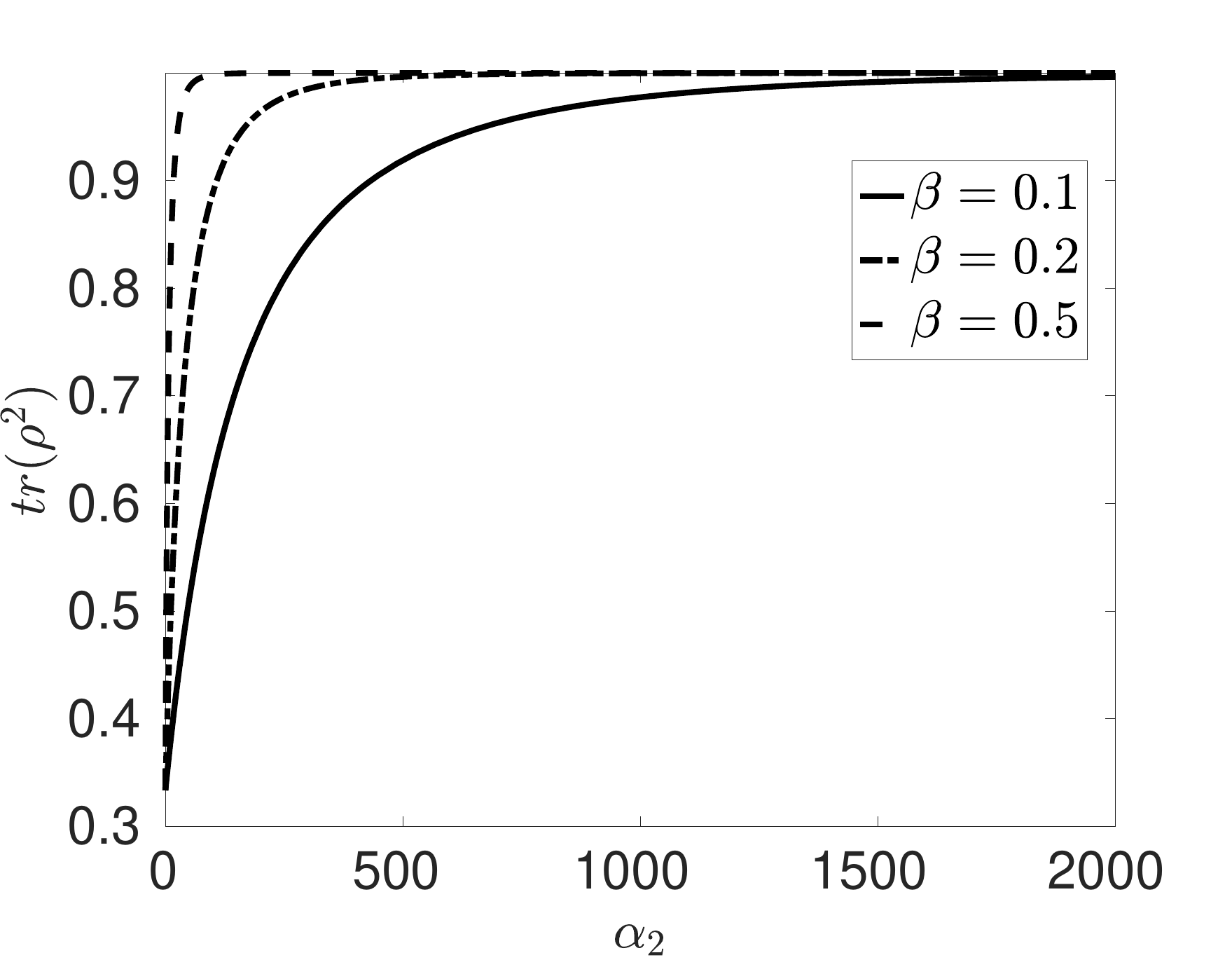}}
	\subfigure[]{\includegraphics[width=8cm]{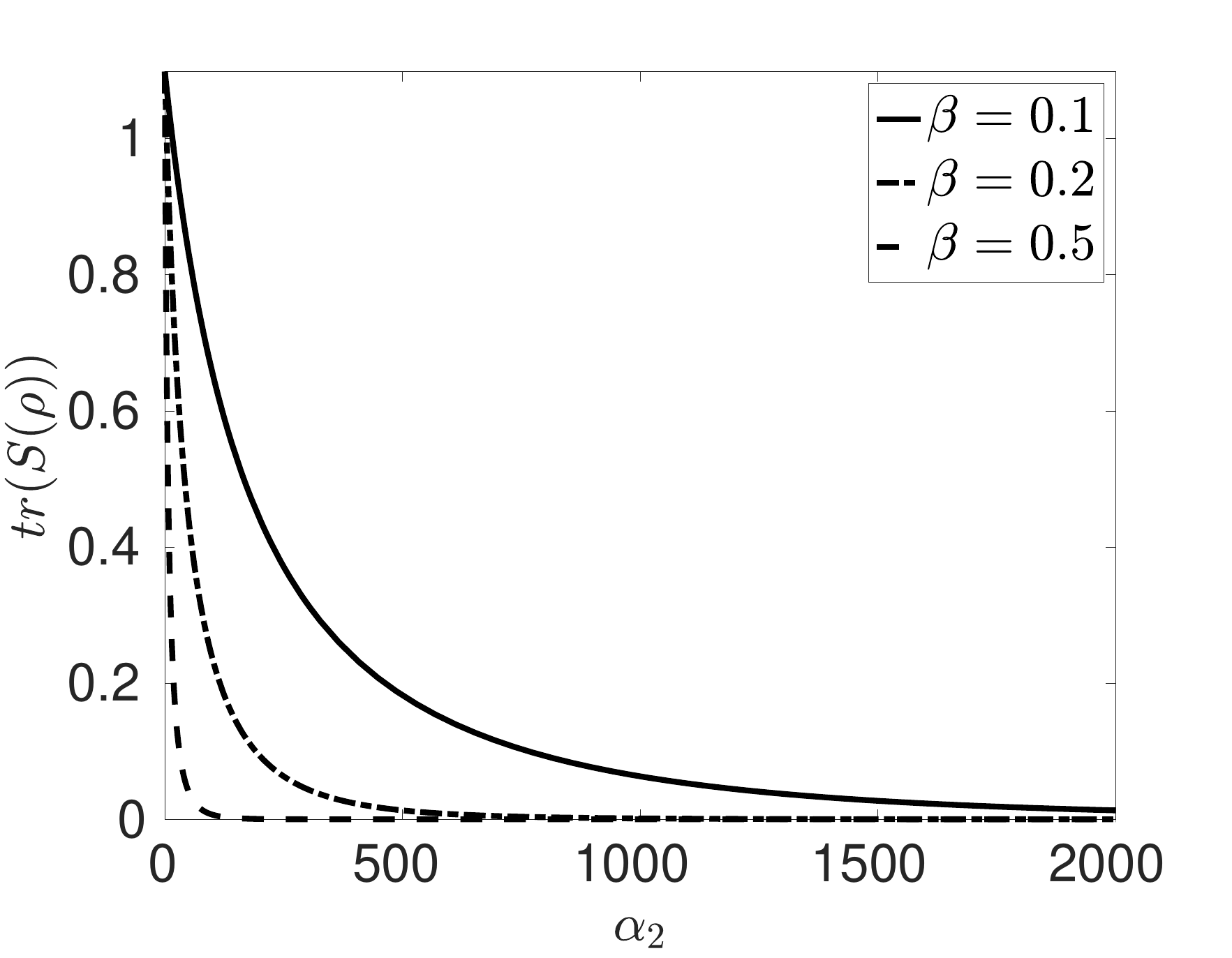}}
	\caption{\textbf{(a)} Behavior of the purity $tr(\rho^2)$ for various value of $\beta$ and with $\alpha>-1/2$, $\alpha_1=1$ for the RDM II example. As $\alpha_2\rightarrow -1/2$, $\rho$ tends to a Riesz pure state, whereas for $\alpha_2\rightarrow \infty$, $\rho$ is a fully mixed state. \textbf{(b)}) Same as \textbf{(a)} but for the entropy $tr(S(\rho))$.}\label{fig_unbroken}
\end{figure}	

\subsubsection{GDM in time independent case}\label{sectGDM}
We now focus on the possibility of obtaining a generalized density matrix (GDM) by considering a deformation matrix $R$ that is not invertible and satisfies condition \eqref{213}. When $1 + 2 \alpha_1 \alpha_2 = 0$, and maintaining only $\alpha_1$ as main parameter, the previous deformation matrix $R$ is not invertible, and has the following form
\be\label{R*}
R = \begin{pmatrix}
	0 & \sqrt{2} \alpha_1 &  \sqrt{2} \alpha_1 \\
	1 & 0 & 0 \\
	0 & 1 & 1
\end{pmatrix}.
\en
We notice that, introducing as in Section \ref{sect-IO}, $\varphi_j=Re_j$, $\varphi_2$ and $\varphi_3$ are proportional one to the other. Hence $\F_\varphi$ cannot be a basis of $\Hil=\mathbb{C}^3$, but it is still possible to use $\varphi_1$ and $\varphi_2$ to generate $\Hil_\varphi$, which is essentially $\mathbb{C}^2$.
It is clear that further constraint on $\rho_0$ must be taken into account to  fulfill \eqref{213}. Selecting again a diagonal form $\rho_0 = \sum_j \lambda_j |e_j\rangle \langle e_j|$ with  $\lambda_3=\lambda_2=1/2-\lambda_1/2$, which are different from those considered so far, one can check that \eqref{213} is satisfied by taking 
\be\label{rhogdm}
\rho = \begin{pmatrix}
	(1-\lambda_1)/2 & 0 &  0 \\
	0 & \lambda_1 & 0 \\
	(1-\lambda_1)/2\sqrt{2}\alpha_1 & 0 & 0
\end{pmatrix}.
\en
We observe that $tr(\rho)=\frac{1+\lambda_1}{2}\neq1$, in general. However we also see that (\ref{214}) is satisfied. This simple example shows that a GDM could easily have a trace which is not one.
The above choice allows also to satisfy the intertwining condition \eqref{220} where the entropy operator for $\rho$ is given by
$$
S(\rho)=\left(
\begin{array}{ccc}
	-\left(\frac{1}{2} - \frac{\lambda_1}{2}\right) \log\left(\frac{1}{2} - \frac{\lambda_1}{2}\right) & 0 & 0 \\
	0 & -\lambda_1 \log(\lambda_1) & 0 \\
	-\frac{\left(\frac{1}{2} - \frac{\lambda_1}{2}\right) \log\left(\frac{1}{2} - \frac{\lambda_1}{2}\right)}{\sqrt{2} \alpha_1} & 0 & 0 \\
\end{array}
\right)
$$
  with trace $
  -\left(\frac{1}{2} - \frac{\lambda_1}{2}\right) \log\left(\frac{1}{2} - \frac{\lambda_1}{2}\right) - \lambda_1 \log(\lambda_1).
  $
We emphasize that the case  $\lambda_1\rightarrow 1$ is to be considered singular, in the sense that   $\rho_0=|e_1\rangle \langle e_1|$ and  $\rho=|e_2\rangle \langle e_2|$ so that $R$ is basically an intertwining operator between the pure states represented by $|e_1\rangle$ and $|e_2\rangle$. In this case the purity and entropy are minimal/maximal, respectively, as shown in Figure \ref{f3dgdm}.

	\begin{figure}[!ht]
\begin{center}
	\includegraphics[width=8.25cm]{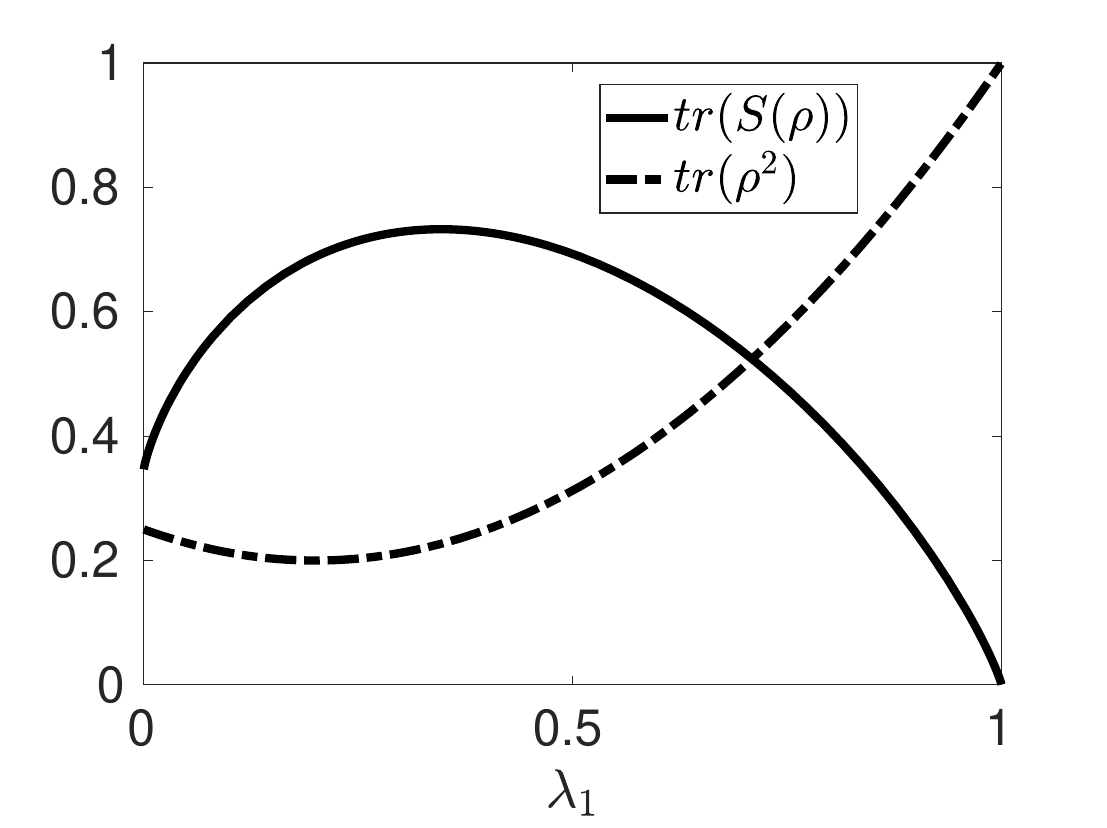}
	\caption{Behavior of purity $tr(\rho^2)$ and entropy $tr(S(\rho))$ as function of the parameter $\lambda_1$ for the GDM example .
	}\label{f3dgdm}
\end{center}
\end{figure}

Before ending this section it might be useful to notice that, since our system lives in $\mathbb{C}^3$, and since $\det(R)=0$, $R$ being the matrix in (\ref{R*}), it follows that $\det(R^\dagger R)=\det(RR^\dagger)=0$. Hence our $R$ has not PI, as we have already observed after Definition \ref{defPI}. For this reason, it is not possible to use (\ref{217}) in the present context.

\section{Conclusions}\label{sect-concl}

In this paper we have proposed some natural extensions of the notions of density matrix, pure state and entropy operators. Our main aim was to use our proposals in connection with non-Hermitian quantum mechanics. In particular we have used a deformation which might appear {\em simple}, introducing new operators which are similar to a {\em standard} DM. These are our RDM. Next we have seen what happens, and what can be done, in case of GDMs, i.e. when the similarity map is replaced by an intertwining operator which is not invertible. Our general results are described in two different, finite-dimensional, models. It is particularly interesting to us to remark that, while RDMs share many of the original properties of the DMs they are similar to, the same is not true for GDMs. In fact, already for the simple example in Section \ref{sectGDM} we have seen that the unity of the trace is lost. This, of course, open the way to many questions, and in particular to the concrete physical relevance of GDM. A deeper understanding of this particular aspect is among our future plans. However, intertwining operators have already proved to be interesting in quantum mechanics, and for this reason we are confident that GDMs could have some role in the analysis of some concrete system. In particular the possibility of using (\ref{217}) was not considered here in the examples. We will analyze this possibility in a future paper, in connection with some model defined on some infinitely-dimension Hilbert space.

\section*{Acknowledgements}
{
F.B. and F.G. acknowledge support under the National Recovery and Resilience Plan (NRRP) funded by the European Union - NextGenerationEU -
Project Title "Transport phonema in low dimensional structures: models, simulations and theoretical aspects" - project code 2022TMW2PY - CUP B53D23009500006.	
F.B. and F.G. also  acknowledge the support of the FFR2023-FFR2024 grant of the University of Palermo. L.S. acknowledges financial support from \textit{Progetto REACTION “first and euRopEAn siC eighT Inches pilOt liNe”}. All authors acknowledge partial financial support from G.N.F.M. of the INdAM.}

\end{document}